\documentclass[pra,onecolumn]{revtex4}
\usepackage{amsmath,amsfonts,amsthm,amssymb}
\usepackage{mathtools}
\usepackage{graphicx}
\usepackage{xcolor}
\usepackage{setspace,caption}
\usepackage{braket}

\newcounter{thm}
\theoremstyle{definition}
\newtheorem{defi}[thm]{Definition}

\theoremstyle{plain}
\newtheorem{prop}[thm]{Proposition}
\newtheorem{theo}[thm]{Theorem}
\newtheorem{lemm}[thm]{Lemma}
\newtheorem{coro}[thm]{Corollary}

\begin{document}

    \title{\vspace{-1cm} Testing the Structure of Multipartite Entanglement with Hardy's Nonlocality}
    \author{Lijinzhi Lin}\email{linljz16@mails.tsinghua.edu.cn}
    \author{Zhaohui Wei}\email{weizhaohui@gmail.com}
    \affiliation{Center for Quantum Information, Institute for Interdisciplinary Information Sciences, Tsinghua University, Beijing 100084, P. R. China}

	\begin{abstract}
Multipartite quantum states may exhibit different types of quantum entanglement in that they cannot be converted into each other by local quantum operations only, and fully understanding mathematical structures of different types of multipartite entanglement is a very challenging task. In this paper, from the viewpoint of Hardy's nonlocality, we compare W and GHZ states and show a couple of crucial different behaviors between them. Particularly, by developing a geometric model for the Hardy's nonlocality problem of W states, we derive an upper bound for its maximal violation probability, which turns out to be strictly smaller than the corresponding probability of GHZ state. This gives us a new comparison between these two quantum states, and the result is also consistent with our intuition that GHZ states is more entangled. Furthermore, we generalize our approach to obtain an asymptotic characterization for general $N$-qubit W states, revealing that when $N$ goes up, the speed that the maximum violation probabilities decay is exponentially slower than that of general $N$-qubit GHZ states. We provide some numerical simulations to verify our theoretical results.
	\end{abstract}

    \maketitle

	\begin{section}{Introduction}

		Entanglement plays a central role in quantum information processing tasks, and it is often entanglement that makes quantum schemes enjoy
        remarkable advantage over their classical counterparts. Therefore,
        studying and characterizing the properties of quantum entanglement is naturally an important and fundamental problem. At present, the structure of quantum entanglement for bipartite quantum states has been relatively clear, especially the case of pure states. However, the situation of multipartite entanglement is much more complicated, and it is still far from being understood very well. Nevertheness, a remarkable fact on multipartite entanglement has been well-known, that is, multipartite quantum states can be entangled in different ways, in that different kinds of multipartite entanglement can not be converted into each other by local operations only \cite{DVC00}. A most famous example that demonstrates this fact is Greenberger-Horne-Zeilinger (GHZ) and W states, as they are two different forms of entanglement in three-qubit quantum states \cite{DVC00}.

        Different entanglement forms exhibit different properties. In the example of GHZ and W states, it has been well-known that GHZ state is more entangled, but W state is more robust against qubit loss. For general case of multipartite entanglement, however, very little like this is known. In order to gain a deep understanding of this problem, characterizing different entanglement forms from more viewpoints are highly demanded.

        One attempt of this kind is comparing the underlying quantum nonlocality given by different forms of multipartite entanglement, and this
        approach allows us to observe their differences directly in quantum labs \cite{BSV12}. What is more, since the differences come from nonlocality, this actually provides us a device-independent way to achieve the task, making it possible to distinguish different kinds of entanglement reliably by using unreliable quantum devices. In fact, it has been shown that Bell inequalities exist such that they can be violated by W states but not by GHZ states, and vice versa \cite{BSV12}. Based on observed quantum nonlocality, a lot of interesting results that certify the existence of multipartite entanglement have also been reported \cite{CGP+02,BGL+11,ZDBS19,ATB+19}.

        Besides Bell inequalities, Hardy's paradox provides another framework to describe quantum nonlocality \cite{Hardy93}. For convenience, in later discussions we use Hardy's nonlocality to address the nonlocal property revealed by Hardy's paradox.

        The original Hardy's nonlocality problem was a proof of entanglement for almost all two qubit states \cite{Hardy93}, and later was generalized to scenarios of multiple qubits, multiple settings, and qudit states \cite{Hardy97,Cereceda04,Chen13,Chen18}. Furthermore, a lot of experiments have been performed to confirm the paradox \cite{TBMM95,LS09,Karimi14}. In this paper, our comparisons will be based on the Hardy's nonlocality problem for multiqubit states proposed in \cite{Cereceda04}, which can be formulated as below. Consider an $N$-qubit quantum state $\ket{\psi}$ and two sets of observables $U_i$ and $D_i$($i\in [N]$, where $[N]\equiv\{1,2,...,N\}$), where the subscript $i$ represents that the observable measures the $i$-th qubit alone. The observables are set up so that
		\begin{align*}
			\mathrm{P}(D_1 U_2 \cdots U_N|\mathrm{++}\cdots\mathrm{+})= & 0, \\
			\mathrm{P}(U_1 D_2 \cdots U_N|\mathrm{++}\cdots\mathrm{+})= & 0, \\
			\cdots & \\
			\mathrm{P}(U_1 U_2 \cdots D_N|\mathrm{++}\cdots\mathrm{+})= & 0, \\
			\mathrm{P}(D_1 D_2 \cdots D_N|\mathrm{--}\cdots\mathrm{-})= & 0, \\
			\mathrm{P}(U_1 U_2 \cdots U_N|\mathrm{++}\cdots\mathrm{+})> & 0,
		\end{align*}
		where $\mathrm{P}(A_1 A_2 \cdots A_N|\mathrm{+++})$ denotes the joint probability when one measures the $i$-th qubit with the measurement setting $A_i$ and gets the outcome $+$, and the other expressions are similar. And for convenience, we call the first $N$ relations \emph{equation constraints}.

		It turns out that quantum entanglement is necessary to manifest Hardy's nonlocality, i.e., satisfy all the constraints above \cite{Hardy93,Cereceda04}. Indeed, in any classical scenario where each local measurement is independent of the others, the last inequality gives $\mathrm{P}(U_i|\mathrm{+})>0$ for all $i\in[N]$, implying that $\mathrm{P}(D_i|+)=0$ for all $i\in[N]$, which is a contradiction to $\mathrm{P}(D_1 D_2 \cdots D_N|\mathrm{--}\cdots\mathrm{-})=0$.
Therefore, if a classical system satisfies the first $N$ constraints, we must have that $\mathrm{P}(U_1 U_2 \cdots U_N|\mathrm{++}\cdots\mathrm{+})=0$, and the violation to this relation means that the system must be quantum. For convenience, when the first $N$ constraints are satisfied, we call the maximal value of $\mathrm{P}(U_1 U_2 \cdots U_N|\mathrm{++}\cdots\mathrm{+})$ the maximal violation probability.

		Since Hardy's nonlocality reveals quantumness of entangled quantum states, it should allow us to look into the essential properties of
        multipartite entanglement, including describing the differences between multipartite entanglement forms. However, to our knowledge Hardy's nonlocality has not been utilized to compare different entanglement structures of multipartite quantum states. In this paper, complementing a previous work that investigated Hardy's nonlocality for multipartite GHZ states \cite{Cereceda04}, we analyze Hardy's nonlocality for multipartite W states.
        
        Specifically, by developing a new geometric model for W states in Hardy's nonlocality problem, we derive an upper bound of the maximal violation probability for the perfect 3 qubit W state, which is 1/9 and strictly smaller than the corresponding probability of the perfect 3 qubit GHZ state, 0.125. Note that this comparison is consistent with our intuition that the GHZ state is more entangled, though we have known that entanglement and nonlocality are two different computational resources. Furthermore, we also obtain an asymptotic lower bound of maximum violation probabilities for multipartite W states as well, which is roughly $\Omega(1/N)$. And this means that when $N$ goes up, the speed that maximum violation probabilities for multipartite W states decay is exponentially slower than that of multipartite GHZ states. Therefore, our results indicate a couple of crucial different behaviors of W states and GHZ states from the viewpoint of Hardy's nonlocality. We also provide some numerical simulation results to verify our theoretical results.

%



	\end{section}

	\begin{section}{Geometric Model for Generalized 3 Qubit W States}

		In this section, we first consider the \emph{generalized W state}, which can be expressed as
		\begin{align}\label{eq:generalizedW}
			\ket{\psi}= & a_1\ket{100}+a_2\ket{010}+a_3\ket{001},
		\end{align}
		where $a_i\neq0$ and $\sum_i |a_i|^2=1$.
		By applying local phases on the basis state $\ket{1}$ of each qubit, we may assume without loss of generality that $a_i>0$. 

        Note that the constraints and the objective in Hardy's nonlocality problem are given by relations on joint probability distributions of measurement outcomes of observables. We now develop a
        geometric model to represent local observables for generalized W states, which allows us to formulate these joint probability distributions in the language of vectors. Later we will see that the geometric model can be generalized to $N$-qubit generalized W states.

		Given an observable $A$, let $\lambda$ be one of its eigenvalues with 1-dimensional eigenspace, and its corresponding eigenstate can be
        written as
		\begin{align}\label{eq:state}
			\ket{\phi}=\cos\varphi\ket{0}+e^{i\theta}\sin\varphi\ket{1},
		\end{align}
		where $\varphi\in[0,\pi/2]$ and $\theta\in[0,2\pi)$.
		To build our geometric model, when $\varphi\neq\pi/2$ we make the following definition.

		\begin{defi}
			\label{defi:vec}
			The \emph{representation vector} of the observable/eigenvalue pair $(A,\lambda)$ is defined as
			\begin{align}\label{eq:representation}
				v(A,\lambda)\equiv(\tan\varphi \cos\theta,\tan\varphi \sin\theta)^T\in\mathbb{R}^2,
			\end{align}
            and for convenience, when $\varphi\neq\pi/2$ we say $v(A,\lambda)$ is \emph{well-defined}.
		\end{defi}

		
Recall that Hardy's nonlocality problem is a maximization problem among local observables $U_i$ and $D_i$ with eigenvalues $\pm 1$, where the subscript $i\in[3]$ indicates the observable measuring the $i$-th qubit:
		\begin{alignat}{2}
			& \text{maximize:} & \qquad & \mathrm{P}(U_1 U_2 U_3|\mathrm{+++}),  \label{eq:3Hardy1}\\
			& \text{subject to:} & & \mathrm{P}(D_1 U_2 U_3|\mathrm{+++})=0,  \label{eq:3Hardy2}\\
			& & & \mathrm{P}(U_1 D_2 U_3|\mathrm{+++})=0, \label{eq:3Hardy3}\\
			& & & \mathrm{P}(U_1 U_2 D_3|\mathrm{+++})=0, \label{eq:3Hardy4}\\
			& & & \mathrm{P}(D_1 D_2 D_3|\mathrm{---})=0.  \label{eq:3Hardy5}
		\end{alignat}

		By our geometric model, the above conditions can be restated, as showed in the following proposition.

		\begin{prop}
			\label{prop:prob}
			Let $A_1,A_2,A_3$ be observables with all their eigenspaces being 1-dimensional.
			Let $\lambda_1, \lambda_2, \lambda_3$ be eigenvalues corresponding to $A_1,A_2,A_3$, respectively. Suppose $v(A_i,\lambda_i)$ is well-defined for $i\in[3]$.
			Let $t_i=a_i\cdot v(A_i,\lambda_i)$.
			Then
			\begin{align}\label{eq:joint}
				\mathrm{P}(A_1A_2A_3|\lambda_1\lambda_2\lambda_3)=\frac{\lVert t_1+t_2+t_3 \rVert^2}{(1+\frac{1}{a_1^2}\lVert t_1 \rVert^2)(1+\frac{1}{a_2^2}\lVert t_2 \rVert^2)(1+\frac{1}{a_3^2}\lVert t_3 \rVert^2)}.
			\end{align}
		\end{prop}

		\begin{proof}
			Suppose the eigenstate for $(A_i,\lambda_i)$ pair is
			\begin{align*}
				\ket{\phi_i}=\cos\varphi_i\ket{0}+e^{i\theta_i}\sin\varphi_i\ket{1},
			\end{align*}
			where $\varphi_i\in[0,\pi/2)$ and $\theta_i\in[0,2\pi)$.

			By the postulate of quantum measurement, we have that
			\begin{align*}
				\mathrm{P}(A_1A_2A_3|\lambda_1\lambda_2\lambda_3)=\bra{\psi}(\ket{\phi_1}\bra{\phi_1}\otimes\ket{\phi_2}\bra{\phi_2}\otimes\ket{\phi_3}\bra{\phi_3})\ket{\psi}=a^T Q a,
			\end{align*}
			where $Q$ is a positive semidefinite matrix defined as
			\begin{align}\label{eq:q}
				Q=
				\begin{pmatrix}
					1 & \cos(\theta_1-\theta_2) & \cos(\theta_1-\theta_3) \\
					\cos(\theta_1-\theta_2) & 1 & \cos(\theta_2-\theta_3) \\
					\cos(\theta_1-\theta_3) & \cos(\theta_2-\theta_3) & 1 \\
				\end{pmatrix},
			\end{align}
			and $a$ is a vector defined as
			\begin{align}\label{eq:a}
				a=
				\begin{pmatrix}
					a_1\sin\varphi_1\cos\varphi_2\cos\varphi_3 \\
					a_2\cos\varphi_1\sin\varphi_2\cos\varphi_3 \\
					a_3\cos\varphi_1\cos\varphi_2\sin\varphi_3 \\
				\end{pmatrix}.
			\end{align}

			The matrix $Q$ admits a factorization $Q=B^TB$, where
			\begin{align*}
				B=
				\begin{pmatrix}
					\cos\theta_1 & \cos\theta_2 & \cos\theta_3 \\
					\sin\theta_1 & \sin\theta_2 & \sin\theta_3 \\
				\end{pmatrix}.
			\end{align*}
			Therefore, the factorization gives
			\begin{align}\label{eq:Ba}
				\mathrm{P}(A_1A_2A_3|\lambda_1\lambda_2\lambda_3)=\lVert Ba \rVert^2.
			\end{align}

			Extracting the factor $\prod_i \cos^2 \varphi_i$ from the outcome probability, we have that
			\begin{align*}
				\mathrm{P}(A_1A_2A_3|\lambda_1\lambda_2\lambda_3)=\left(\prod_i \cos^2 \varphi_i \right)\left\lVert B
				\begin{pmatrix}
					a_1 \tan\varphi_1 \\
					a_2 \tan\varphi_2 \\
					a_3 \tan\varphi_3 \\
				\end{pmatrix}
				\right\rVert^2.
			\end{align*}
			According to the definition of $v(A_i,\lambda_i)$, it holds that
			\begin{align*}
				\mathrm{P}(A_1A_2A_3|\lambda_1\lambda_2\lambda_3)=\left(\prod_i \cos^2 \varphi_i \right)\lVert t_1+t_2+t_3\rVert^2.
			\end{align*}
			In the meanwhile, the cosine factors can be rewritten as
			\begin{align*}
				\cos^2 \varphi_i=\frac{1}{1+\tan^2\varphi_i}=\frac{1}{1+\frac{1}{a_i^2}\lVert t_i\rVert^2},
			\end{align*}
			which means that
			\begin{align*}
				\mathrm{P}(A_1A_2A_3|\lambda_1\lambda_2\lambda_3)=\frac{\lVert t_1+t_2+t_3 \rVert^2}{(1+\frac{1}{a_1^2}\lVert t_1 \rVert^2)(1+\frac{1}{a_2^2}\lVert t_2 \rVert^2)(1+\frac{1}{a_3^2}\lVert t_3 \rVert^2)}.
			\end{align*}
			This concludes the proof.
		\end{proof}

		We immediately have the following corollary:
		\begin{coro}
			\label{coro:annihilate}
			Let $A_1,A_2,A_3$ be observables with all their eigenspaces being 1-dimensional, and let $\lambda_1, \lambda_2, \lambda_3$ be eigenvalues corresponding to $A_1,A_2,A_3$, respectively. Suppose $v(A_i,\lambda_i)$ is well-defined for $i\in[3]$. Let $t_i=a_i\cdot v(A_i,\lambda_i)$.
			Then
			\begin{align*}
				\mathrm{P}(A_1A_2A_3|\lambda_1\lambda_2\lambda_3)=0
			\end{align*}
            if and only if
            \begin{equation}\label{eq:sum0}
            t_1+t_2+t_3=0.
            \end{equation}
		\end{coro}

		In order to formulate all constraints in Hardy's nonlocality problem, we make the following further definitions.
		\begin{defi}
			\label{defi:hardyvec}
			For $i\in[3]$, if $v(U_i,+1)$, $v(D_i,+1)$, and $v(D_i,-1)$ are well-defined, let
			\begin{align}\label{eq:vectors}
				u_i= & a_i v(U_i,+1), \\
				v_i= & a_i v(D_i,+1), \\
				w_i= & a_i v(D_i,-1).
			\end{align}
		\end{defi}
		
		With the new notations, we now translate the constraints in Hardy's nonlocality problem in the language of vectors defined above. First, by Corollary \ref{coro:annihilate}, we have
        \begin{equation}\label{eq:UUD}
            v_i=-\sum_{j\neq i}u_j
        \end{equation}
        for $i\in[3]$, and
        \begin{equation}\label{eq:DDD}
            \sum_{j}w_j=0.
        \end{equation}

    	Second, by Definition \ref{defi:vec}, we have $v_i=-\frac{a_i^2 w_i}{\lVert w_i \rVert^2}$. Indeed, suppose the eigenstate for $(D_i,+1)$ is
				\begin{align*}
					\ket{\phi_{+}}=\cos\alpha_i\ket{0}+e^{i\beta_i}\sin\alpha_i\ket{1}.
				\end{align*}
				Then the eigenstate for $(D_i,-1)$ is
				\begin{align*}
					\ket{\phi_{-}}=\sin\alpha_i\ket{0}-e^{i\beta_i}\cos\alpha_i\ket{1}.
				\end{align*}
				Now, by Definition \ref{defi:vec}, we have that
				\begin{align*}
					v_i= & a_i(\tan\alpha_i \cos\beta_i,\tan\alpha_i \sin\beta_i)^T, \\
					w_i= & -a_i(\cot\alpha_i \cos\beta_i,\cot\alpha_i \sin\beta_i)^T,
				\end{align*}
				hence $v_i=-\frac{a_i^2 w_i}{\lVert w_i \rVert^2}$.
		
		With the above observations, when all representation vectors are well-defined, the probability maximization problem can be rewritten as:
		\begin{alignat*}{2}
			& \text{maximize:} & \qquad & \mathrm{P}(U_1 U_2 U_3|\mathrm{+++}) \\
			& & = & \frac{\lVert u_1+u_2+u_3 \rVert^2}{(1+\frac{1}{a_1^2}\lVert u_1 \rVert^2)(1+\frac{1}{a_2^2}\lVert u_2 \rVert^2)(1+\frac{1}{a_3^2}\lVert u_3 \rVert^2)} \\
			& & = & \frac{\lVert v_1+v_2+v_3 \rVert^2/4}{(1+\frac{1}{4a_1^2}\lVert v_2+v_3-v_1 \rVert^2)(1+\frac{1}{4a_2^2}\lVert v_3+v_1-v_2 \rVert^2)(1+\frac{1}{4a_3^2}\lVert v_1+v_2-v_3 \rVert^2)}, \\
			& \text{subject to:} & & w_1+w_2+w_3=0,
		\end{alignat*}
		where $w_i$($i\in[3]$) are the variables and
		\begin{align*}
			v_i=-\frac{a_i^2 w_i}{\lVert w_i\rVert^2}.
		\end{align*}
		
	\end{section}

	\begin{section}{Bounding the Violation Probability for the W State}
		
		Based on the geometric model introduced above, we now prove our first main result, which shows that for the perfect W state the maximal violation probability in the Hardy's nonlocality problem is upper bound for 1/9. Since the geometric model supposes all representation vectors are well-defined, we first consider this case, then we show that the conclusion can be generalized to arbitrary case.
		\begin{lemm}
			\label{lemm:bound}
			Let $U_i$ and $D_i$ be observables in the Hardy's nonlocality problem for the W state $\ket{\psi}=1/\sqrt{3}(\ket{001}+\ket{010}+\ket{100})$.
			If $v(U_i,+1)$, $v(D_i,+1)$, and $v(D_i,-1)$ are well-defined for $i\in[3]$,
			then all equation constraints are satisfied implies that
			\begin{align*}
				\mathrm{P}(U_1 U_2 U_3|\mathrm{+++})\leq 1/9.
			\end{align*}
		\end{lemm}

		\begin{proof}
			We have known that the target probability can be expressed as
			\begin{align*}
				& \mathrm{P}(U_1U_2U_3|\mathrm{+++}) \\
				= & \frac{\frac{1}{4}\lVert v_1+v_2+v_3\rVert^2}
				{
					\left(1+\frac{3}{4}\lVert -v_1+v_2+v_3\rVert^2\right)
					\left(1+\frac{3}{4}\lVert v_1-v_2+v_3\rVert^2\right)
					\left(1+\frac{3}{4}\lVert v_1+v_2-v_3\rVert^2\right)
				}.
			\end{align*}
			Expanding the denominator gives
			\begin{align*}
				& \mathrm{P}(U_1U_2U_3|+++) \\
				\leq &
				\frac{\frac{1}{4}\lVert v_1+v_2+v_3\rVert^2}
				{
					1+\frac{3}{4}\left(
					\lVert -v_1+v_2+v_3\rVert^2+
					\lVert v_1-v_2+v_3\rVert^2+
					\lVert v_1+v_2-v_3\rVert^2
					\right)
				} \\
				= &
				\frac{\frac{1}{4}\lVert v_1+v_2+v_3\rVert^2}
				{
					1+\frac{9}{4}\lVert v_1+v_2+v_3\rVert^2
					-6\left(v_1\cdot v_2+v_2\cdot v_3+v_3\cdot v_1\right)
				}.
			\end{align*}

			Since isometries preserve inner products, we may assume that
			\begin{align*}
				w_1= & -(M,0)^T, \\
				w_2= & -(x,y)^T, \\
				w_3= & -(-x-M,-y)^T,
			\end{align*}
			without loss of generality, where we have utilized the relation $w_1+w_2+w_3=0$.
			By the assumption that all $v(D_i,+1)$ and $v(D_i,-1)$ are well-defined,
			we have $w_1, w_2, w_3\neq 0$; that is, $M\neq 0$, $x^2+y^2\neq 0$ and $(x+M)^2+y^2\neq 0$.

			By relation $v_i=-w_i/(3\lVert w_i\rVert^2)$,
			we have
			\begin{align*}
				v_1= & (1/M,0)^T/3, \\
				v_2= & \left(\frac{x}{x^2+y^2},\frac{y}{x^2+y^2}\right)^T/3, \\
				v_3= & \left(\frac{-x-M}{(x+M)^2+y^2},\frac{-y}{(x+M)^2+y^2}\right)^T/3.
			\end{align*}
			Then
			\begin{align*}
				& v_1\cdot v_2+v_2\cdot v_3+v_3\cdot v_1 \\
				= & \frac{\frac{x}{M}\left((x+M)^2+y^2\right)-x(x+M)-y^2-\frac{x+M}{M}\left(x^2+y^2\right)}
				{9\left(x^2+y^2\right)(\left(x+M)^2+y^2\right)} \\
				= & \frac{-2y^2}
				{9\left(x^2+y^2\right)(\left(x+M)^2+y^2\right)} \leq 0.
			\end{align*}
			Therefore,
			\begin{align*}
				\mathrm{P}(U_1U_2U_3|\mathrm{+++}) \leq \frac{\frac{1}{4}\lVert v_1+v_2+v_3\rVert^2}{1+\frac{9}{4}\lVert v_1+v_2+v_3\rVert^2} \leq 1/9.
			\end{align*}
		\end{proof}

		We now show that the assumption in Lemma \eqref{lemm:bound} that all representation vectors involved in the Hardy's nonlocality problem are well-defined can be removed, which means that in this case the upper bound in Lemma \eqref{lemm:bound} is still correct.

        For this, we first suppose two of $v(U_i,+1)$ are not well-defined, then it can be seen that the vector $a$ in Eq.\eqref{eq:a} for $P(U_1U_2U_3|+++)$ is zero, thus Eq.\eqref{eq:Ba} indicates that $P(U_1U_2U_3|+++)=0$, and it does not hurt the upper bound.
        Second, similar argument shows that if one of $v(A_i|\lambda_i)$ is not well-defined, then $P(A_1A_2A_3|\lambda_1\lambda_2\lambda_3)=0$ implies that there must be another $i'\neq i$ such that $v(A_{i'}|\lambda_{i'})$ is not well-defined either.

        Then combining the above two observations, we can rule out the possibility that only one of $v(U_i,+1)$, say $v(U_1,+1)$, is not well-defined. If this is the case, then Eq.\eqref{eq:3Hardy3} and Eq.\eqref{eq:3Hardy4} means that $v(D_2,+1)$ and $v(D_3,+1)$ are not well-defined, i.e., $v(D_2,-1)=0$ and $v(D_3,-1)=0$. By applying Corollary \ref{coro:annihilate} on $P(D_1D_2D_3|---)=0$, we have that $v(D_1,-1)=0$, and this indicates that $v(D_1,+1)$ is not well-defined either. However, we know that $P(D_1U_2U_3|+++)=0$, and this needs that at least one of $v(U_2,+1)$ and $v(U_3,+1)$ is not well-defined, a contradiction. In summary, if any vector in $v(U_i,+1)$, $v(D_i,+1)$, and $v(D_i,-1)$ is not well-defined, satisfying all equation constraints means that $P(U_1U_2U_3|+++)=0$. Therefore, we have the following theorem.

%
%

		\begin{theo}\label{theo:bound}
			Let $U_i$ and $D_i$ be observables in the Hardy's nonlocality problem for the W state $\ket{\psi}=1/\sqrt{3}(\ket{001}+\ket{010}+\ket{100})$. Then all equation constraints are satisfied implies that
			\begin{align*}
				\mathrm{P}(U_1 U_2 U_3|\mathrm{+++})\leq 1/9.
			\end{align*}
		\end{theo}

%
%
%
%

		Theorem \ref{theo:bound} essentially states that the violation probability of the perfect 3 qubit W state is upper-bounded by $1/9$.
		For comparison, it has been shown that obtaining a violation probability of $0.125$ is possible from the perfect 3 qubit GHZ state \cite{Cereceda04}.
		
	\end{section}

	\begin{section}{Generalization to $N$ Qubit W States}

		In this section, our first task is to show that the geometric model introduced above can be generalized to $N$ qubit W states with $N>3$.

		The $N$ qubit generalized $W$ state is defined as
		\begin{align*}
			\ket{\psi}= & a_1\ket{10\cdots 0}+a_2\ket{01\cdots 0}+\cdots+a_N\ket{00\cdots 1},
		\end{align*}
		where $a_i> 0$ for all $i\in[N]$ and $\sum_ia_i^2=1$. When $a_i=1/\sqrt{N}$ for all $i\in[N]$, we call it the $N$ qubit perfect $W$ state, denoted by $\ket{W_N}$. For convenience, we denote the $N$ qubit perfect $GHZ$ state as
\begin{equation}
\ket{GHZ_N}=\frac{1}{2}(\ket{00\cdots 0}+\ket{11\cdots 1}).
\end{equation}

		Following Definition \ref{defi:vec}, the joint measurement outcome probability formula is readily generalized as Proposition \ref{prop:multiprob}.

		\begin{prop}
			\label{prop:multiprob}
			Let $A_i(i\in[N])$ be observables with all their eigenspaces being 1-dimensional.
			Let $\lambda_i$ be an eigenvalue corresponding to $A_i$. Suppose $v(A_i,\lambda_i)$ is well-defined for $i\in[N]$.
			Let $v_i=a_i\cdot v(A_i,\lambda_i)$.
			Then
			\begin{align*}
				\mathrm{P}(A_1\cdots A_N|\lambda_1\cdots \lambda_N)=\frac{\left\lVert \sum\limits_{i=1}^{N}v_i \right\rVert^2}{\prod\limits_{i=1}^{N}\left(1+\frac{1}{a_i^2}\lVert v_i \rVert^2\right)}.
			\end{align*}
		\end{prop}

		\begin{coro}
			\label{coro:multiannihilate}
			Let $A_i(i\in[N])$ be observables with all their eigenspaces being 1-dimensional.
			Let $\lambda_i$ be an eigenvalue corresponding to $A_i$. Suppose $v(A_i,\lambda_i)$ is well-defined for $i\in[N]$.
			Let $v_i=a_i\cdot v(A_i,\lambda_i)$.
			Then
			\begin{align*}
				\mathrm{P}(A_1\cdots A_N|\lambda_1\cdots \lambda_3)=0
			\end{align*}
            if and only if
            \begin{align*}
            \sum\limits_{i=1}^{N}v_i=0.
            \end{align*}
		\end{coro}

		The $N$ qubit Hardy's nonlocality can be restated as the following maximization problem among the observables $U_i$ and $D_i$:
		\begin{alignat*}{2}
			& \text{maximize:} & \qquad & \mathrm{P}(U_1U_2 \cdots U_N|\mathrm{++}\cdots\mathrm{+}), \\
			& \text{subject to:} & & \mathrm{P}(D_1 U_2 \cdots U_N|\mathrm{++}\cdots\mathrm{+})=0, \\
			& & & \mathrm{P}(U_1 D_2 \cdots U_N|\mathrm{++}\cdots\mathrm{+})=0, \\
			& & & \cdots \\
			& & & \mathrm{P}(U_1 U_2 \cdots D_N|\mathrm{++}\cdots\mathrm{+})=0, \\
			& & & \mathrm{P}(D_1 D_2 \cdots D_N|\mathrm{--}\cdots\mathrm{-})=0.
		\end{alignat*}

		\begin{defi}
			\label{defi:mutlihardyvec}
			For $i\in[N]$, if $v(U_i,+1)$, $v(D_i,+1)$, and $v(D_i,-1)$ are well-defined, let
			\begin{align*}
				u_i= & a_i v(U_i,+1), \\
				v_i= & a_i v(D_i,+1), \\
				w_i= & a_i v(D_i,-1).
			\end{align*}
			Additionally, let $u=\sum_{i\in[N]} u_i$, $v=\sum_{i\in[N]} v_i$ and $w=\sum_{i\in[N]} w_i$.
		\end{defi}

		By the constraints in Hardy's nonlocality, we have the following relations:
		\begin{align*}
			\forall i \in[N], \ \  & v_i=-a_i^2w_i/\lVert w_i\rVert^2, \\
			\forall i \in[N], \ \  & v_i=-(u-u_i), \\
			\forall i \in[N], \ \  & u_i=v_i+u, \\
			& u=-\frac{v}{N-1}.
		\end{align*}
		
		Under the relations above, the probability maximization problem become:
		\begin{alignat*}{2}
			& \text{maximize:} & \qquad & \mathrm{P}(U_1 U_2 \cdots U_N|\mathrm{++}\cdots\mathrm{+}) \\
			& & = & \frac{\lVert u \rVert^2}{\prod\limits_{i=1}^{N}\Big(1+\frac{1}{a_i^2}\lVert u_i \rVert^2\Big)} \\
			& & = & \frac{\lVert v \rVert^2/(N-1)^2}{\prod\limits_{i=1}^{N}\Big(1+\lVert v-(N-1)v_i \rVert^2/(((N-1)^2 a_i^2))\Big)}\\
			& \text{subject to:} & & w=0,
		\end{alignat*}
		where $w_i$($i\in[N]$) are the variables and
		\begin{align*}
			v_i=-\frac{a_i^2 w_i}{\lVert w_i\rVert^2}.
		\end{align*}

		We now turn to the second task of this section. Different from the $3$-qubit case, we consider lower bounding the maximum violation probability of the Hardy nonlocality problem when $N$ is large. The following theorem gives such an asymptotic lower bound. Since we are focusing on a lower bound, we can suppose that all the involved representation vectors are well-defined.

		\begin{theo}
			\label{theo:mutliperbound}
			Let $P(N)$ denote the maximum violation probability in the Hardy's nonlocality problem for the perfect $N$ qubit $W$ states.
			Then $P(N)=\Omega\left(N^{-1}\right)$.
		\end{theo}

		\begin{proof}
			For simplicity, we represent the vectors $u_i, v_i, w_i$ with one real number each, in the sense that their second component is equal to zero.

			Let $w_i=1/(N-1)$ for $i\in[N-1]$ and $w_N=-1$.
			Then, $\forall i\in[N-1]$,
			\begin{align*}
				v_i= & -\frac{(N-1)}{N}, \\
				v_N= & \frac{1}{N}, \\
				u= & \frac{N-2}{N-1}, \\
				u_i= & -\frac{1}{N(N-1)}, \\
				u_N= & \frac{N^2-N-1}{N(N-1)}.
			\end{align*}
			Now let $N$ tend to $+\infty$.
			By Proposition \ref{prop:multiprob}, the violation probability under this settings is
			\begin{align*}
				& \frac{\lVert u \rVert^2}{\prod\limits_{i=1}^{N}\left(1+\frac{1}{a_i^2}\lVert u_i \rVert^2\right)} \\
				\approx & \frac{1}{\left(1+\frac{1}{N(N-1)^2}\right)^{N-1}\left(1+\frac{(N^2-N-1)^2}{N(N-1)^2}\right)} \\
				\approx & \frac{1/N}{\left(1+\frac{1}{N(N-1)^2}\right)^{N-1}} \\
				\approx & 1/N.
			\end{align*}
			Therefore, we have $P(N)=\Omega \left(N^{-1}\right)$.

		\end{proof}

		As a comparison, it has been known that the maximum violation probabilities for the $N$ qubit perfect GHZ states diminish exponentially with $N$ \cite{Cereceda04}, thus we witness another sharp difference between asymptotic behaviors of $\ket{W_N}$ and $\ket{GHZ_N}$ when $N$ tends to infinite. Therefore, on one hand $\ket{GHZ_N}$ enjoys stronger nonlocality than $\ket{W_N}$ if $N=3$, but on the other hand, when $N$ becomes larger the speed that Hardy's nonlocality of $\ket{W_N}$ decays is much slower.

	\end{section}

	\begin{section}{Numerical Simulation Results}

    We made the following numerical simulations to verify or complement our theoretical results.

    \begin{subsection}{The $3$-qubit perfect W state}

        We first consider the case of the $3$-qubit perfect W state. In order to parameterize the involved measurements, let
        \begin{align}
				v(U_i,+1)\equiv(\tan\varphi_{1,i} \cos\theta_{1,i},\tan\varphi_{1,i} \sin\theta_{1,i})^T\in\mathbb{R}^2
		\end{align}
        and
        \begin{align}
				v(D_i,+1)\equiv(\tan\varphi_{2,i} \cos\theta_{2,i},\tan\varphi_{2,i} \sin\theta_{2,i})^T\in\mathbb{R}^2
		\end{align}


		For this, we consider the following construction, which is guided by the geometric model.
		\begin{gather}
			\tan\varphi_{11}=M/4, \tan\varphi_{12}=5M/4, \tan\varphi_{13}=M/4, \\
			\theta_{11}=0, \theta_{12}=\pi, \theta_{13}=0, \\
			\tan\varphi_{21}=M, \tan\varphi_{22}=M/2, \tan\varphi_{23}=M, \\
			\theta_{21}=0, \theta_{22}=\pi, \theta_{23}=0,
		\end{gather}
		where $M>0$ is a parameter. In this setting, when $M$ is picked to maximize the violation probability,
		the outcome probability is about $0.071868$, which is indeed below the maximum violation probability $0.125$ that the perfect 3 qubit GHZ state can achieve \cite{Cereceda04}.
    \end{subsection}

    \begin{subsection}{The $3$-qubit generalized W state}
    		For generalized 3 qubit W states,
		when $a_1=0.448473$, $a_2=0.632011$ and $a_3=0.632008$, the following configuration
		\begin{gather*}
			\tan\varphi_{11}=1.320219, \tan\varphi_{12}=0.147611, \tan\varphi_{13}=0.147611, \\
			\theta_{11}=\pi, \theta_{12}=0, \theta_{13}=0, \\
			\tan\varphi_{21}=0.295222, \tan\varphi_{22}=1.172608, \tan\varphi_{23}=1.172607, \\
			\theta_{21}=\pi, \theta_{22}=0, \theta_{23}=0,
		\end{gather*}
		achieves violation probability of $0.0977381$, which is higher than the perfect 3 qubit W state.
		\begin{center}
			\begin{figure}[!ht]
				\centering
				\includegraphics[scale=0.8]{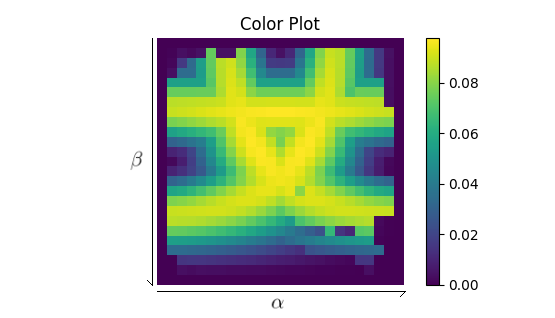}
				\caption{
					The color plot of the maximum violation probability for different amplitude settings $(a_1,a_2,a_3)$.
					The amplitudes are parametrized using spherical coordinate, with angles $\alpha$ and $\beta$ ranging from $0$ to $\pi/2$.
					The angle $\alpha$ increases from left to right (horizontally), while $\beta$ increases from top to bottom (vertically).
				}
				\label{fig:amplitude}
			\end{figure}
		\end{center}
		The results above are obtained via an optimization package.
		Notice that both resulting sequences $(\varphi_{1,i})$ and $(a_i)$ exhibit $S_{N-1}=S_{2}$ symmetry, which matches the W state when the amplitudes are ignored.
		Therefore, we conjecture that the maximum violation probability can be achieved when $a_1=a_2=\cdots=a_{N-1}$ and $\varphi_{1,i}=\varphi_{2,i}=\cdots=\varphi_{N-1,i}$.
		If the conjecture is proven, then the maximization problem would be simplified in the sense that at most two real parameters would be free regardless of $N$.

		Figure \ref{fig:amplitude} is a color plot of the violation probability maximized using optimization package for different amplitudes.
		The horizontal and vertical axes represents $\alpha$ and $\beta$ from $0$ to $\pi/2$, respectively, which are used in the definitions of amplitudes as
		\begin{align*}
			a_1=\cos\beta\cos\alpha, a_2=\cos\beta\sin\alpha, a_3=\sin\beta.
		\end{align*}
		It is evident from the three yellow bands in the color plot that there is decent violation probability.
    \end{subsection}

    \begin{subsection}{The $N$-qubit perfect W state}
    		For the $N$-qubit perfect W state, numerical experiment gives the following lower bounds of the maximum violation probability $P(N)$ for the perfect $N$ qubit W states:
		\begin{center}
			\begin{tabular}{|c|c|}
				\hline
				$N$ & Maximum violation probability \\
				\hline
				3 & 0.07186776197291751 \\
				4 & 0.09802431986561981 \\
				5 & 0.1016666013383646 \\
				6 & 0.0981781711941636 \\
				7 & 0.09256920089757938 \\
				8 & 0.08658662542877839 \\
				9 & 0.08085438836971731 \\
				10 & 0.07557767230678995 \\
				\hline
			\end{tabular}
		\end{center}
		It is evident from the table that $P(N)$ is unimodal in the range $3\leq N\leq10$ and is maximized at $N=5$. This is an intriguing phenomenon
        worth further study. Additionally, the maximum violation probabilities assumed by the perfect $N$ qubit GHZ states are all below $P(N)$ except when $N=3$.
    \end{subsection}

	\end{section}

	\begin{section}{Conclusion}

		In this paper, we have analyzed Hardy's nonlocality for W states.
		For this purpose, we develop a geometric model for general W states, and this model allows us to describe the constraints in Hardy's nonlocality problem as relations on vectors, which in turn makes it convenient to characterize the target violation probability. 

		Concretely, for the perfect $3$-qubit W state, we have shown that its violation probability is upper-bounded by $1/9$.
		As a comparison, the perfect 3 qubit GHZ state has maximum violation probability $0.125$, and the stronger correlation provided by GHZ state is also consistent with our intuition that it is more entangled than W state, though we have known that entanglement and nonlocality are two different computational resources.

		For the perfect $N$-qubit W states where $N\geq 4$, we have shown that their maximum violation probabilities are at least $\Omega(N^{-1})$, making another sharp comparison with the perfect $N$ qubit GHZ states, as the maximum violation probabilities of the latter decay exponentially when $N$ goes up.

Therefore, it can be seen that Hardy's nonlocality indeed provides a new viewpoint to distinguish the different entanglement structures of GHZ and W states. We hope this approach can be generalized to more complicated and more general multipartite quantum states.


		
	\end{section}

\begin{acknowledgments}
L.L. and Z.W. are supported by the National Key R\&D Program of China, Grant No. 2018YFA0306703 and the start-up funds of Tsinghua University, Grant No. 53330100118. This work has been supported in part by the Zhongguancun Haihua Institute for Frontier Information Technology.
\end{acknowledgments}

\end{document}